\documentclass[12pt]{article}

\usepackage{geometry}
\usepackage{cite}
\usepackage{paralist}
\usepackage{graphicx}
\usepackage{subfigure}
\usepackage{amssymb}
\usepackage{amsmath}
\allowdisplaybreaks[1]

\usepackage{amsthm}

\newtheorem{theorem}{Theorem}
\newtheorem{lemma}[theorem]{Lemma}

\makeatletter
\def\@endtheorem{\endtrivlist}
\makeatother

\newcounter{rrule}
\newenvironment{rrule}{\refstepcounter{rrule}\par\smallskip\noindent
\textbf{(R\arabic{rrule})}\quad}{}

\newcounter{brule}
\newenvironment{brule}{\refstepcounter{brule}\par\smallskip\noindent
\textbf{(B\arabic{brule})}\quad}{}
\newcommand{\currentrule}{B\arabic{brule}}

\newcommand{\twins}[1]{\mathrm{Twins}(#1)}

\begin{document}

\title{Algorithms for deletion problems on split graphs}
%\title{Faster branching algorithms for deletion problem on split graphs}
\author{Dekel Tsur%
\thanks{Ben-Gurion University of the Negev.
Email: \texttt{dekelts@cs.bgu.ac.il}}}
\date{}
\maketitle

\begin{abstract}
In the \emph{Split to Block Vertex Deletion} and
\emph{Split to Threshold Vertex Deletion} problems the input is
a split graph $G$ and an integer $k$, and the goal is to decide whether there
is a set $S$ of at most $k$ vertices such that $G-S$ is a block graph and $G-S$
is a threshold graph, respectively.
In this paper we give algorithms for these problems whose running times are
$O^*(2.076^k)$ and $O^*(2.733^k)$, respectively.
\end{abstract}


\section{Introduction}
A graph $G$ is called a \emph{split graph} if its vertex set can be partitioned
into two disjoint sets $C$ and $I$ such that $C$ is a clique and $I$ is an
independent set.
A graph $G$ is a \emph{block graph} if every biconnected component of $G$
is a clique.
A graph $G$ is a \emph{threshold graph} if there is a $t \in \mathbb{R}$
and a function $f \colon V(G) \to \mathbb{R}$ such that
for every $u,v \in V(G)$, $(u,v)$ is an edge in $G$ if and only if
$f(u)+f(v) \geq t$.

In the \emph{Split to Block Vertex Deletion} (SBVD) problem the input is
a split graph $G$ and an integer $k$,
and the goal is to decide whether there is a set $S$ of at most $k$ vertices
such that $G-S$ is a block graph.
Similarly, in the \emph{Split to Threshold Vertex Deletion} (STVD) problem
the input is a split graph $G$ and an integer $k$,
and the goal is to decide whether there is a set $S$ of at most $k$ vertices
such that $G-S$ is a threshold graph.
The SBVD and STVD problems were shown to be NP-hard by Cao et al.~\cite{cao2018vertex}.
A split graph $G$ is a block graph if and only if $G$ does not contain
an induced diamond, where a diamond is a graph with $4$ vertices and $5$ edges.
Additionally, a split graph $G$ is threshold graph if and only if $G$
does not contain an induced path with 4 vertices.
Therefore, SBVD and STVD are special cases of the 4-Hitting Set problem.
Using the fastest known parameterized algorithm for 4-Hitting Set,
due to Fomin et al.~\cite{fomin2010iterative}, the SBVD and STVD problems
can be solved in $O^*(3.076^k)$ time.
Choudhary et al.~\cite{choudhary2019vertex} gave faster algorithms for
SBVD and STVD  whose running times are $O^*(2.303^k)$ and $O^*(2.792^k)$,
respectively.
In this paper we give algorithms for SBVD and STVD whose running times are
$O^*(2.076^k)$ and $O^*(2.733^k)$, respectively.

\section{Preliminaries}
For a graph $G$ and a vertex $v \in V(G)$,
$N(v)$ is the set of vertices that are adjacent to $v$.
For a set $S$ of vertices, $G-S$ is the graph obtained from $G$ by
deleting the vertices of $S$ (and incident edges).
Let $P_4$ denote a graph that is a path on 4 vertices.

In the \emph{3-Hitting Set} problem the input is a family $\mathcal{F}$ of
subsets of size at most 3 of a set $U$ and an integer $k$,
and the goal is to decide whether there is a set $X \subseteq U$
of size at most $k$ such that
$X\cap A \neq \emptyset$ for every $A \in \mathcal{F}$.

For two families of sets $\mathcal{A}$ and $\mathcal{B}$,
$\mathcal{A} \circ \mathcal{B} 
= \{A\cup B \colon A \in \mathcal{A}, B \in \mathcal{B} \}$.

\subsection{Branching algorithm}
A \emph{branching algorithm}  (cf.~\cite{cygan2015parameterized})
for a parameterized problem is a recursive algorithm that uses \emph{rules}.
Given an instance $(G,k)$ to the problem, the algorithm applies some rule.
In each rule, the algorithm either computes the answer to the instance $(G,k)$,
or performs recursive calls on instances
$(G_1,k-c_1),\ldots,(G_t,k-c_t)$, where $c_1,\ldots,c_t > 0$.
The algorithm returns `yes' if and only if at least one recursive call
returned `yes'.
The rule is called a \emph{reduction rule} if $t = 1$, and
a \emph{branching rule} if $t \geq 2$.
To analyze the time complexity of the algorithm, define $T(k)$ to be the
maximum number of leaves in the recursion tree of the algorithm when the
algorithm is run on an instance with parameter $k$.
Each branching rule corresponds to a recurrence on $T(k)$:
\[ T(k) \leq T(k-c_1)+\cdots+T(k-c_t). \]
The largest real root of $P(x) = 1-\sum_{i=1}^t x^{-c_i}$
is called the \emph{branching number} of the rule.
The vector $(c_1,\ldots,c_t)$ is called the \emph{branching vector} of
the rule.

Let $\gamma$ be the maximum branching number over all branching rules.
Assuming that the application of a rule takes $O^*(1)$ time,
the time complexity of the algorithm is $O^*(\gamma^k)$.

\section{Algorithm for SBVD}

\begin{lemma}\label{lem:block}
Let $G$ be a split graph with a partition $C,I$ of its vertices.
$G$ is a block graph if and only if
(1) A vertex in $I$ with degree at least $2$ is adjacent to all vertices in $C$, and
(2) There is at most one vertex in $I$ with degree at least 2.
\end{lemma}
\begin{proof}
Suppose that $G$ is a block graph.
If a vertex $v \in I$ has degree at least 2, let $a_1,a_2 \in C$ be two neighbors
of $v$. For every $b \in C\setminus \{a_1,a_2\}$, we have that $v$ is adjacent
o $b$, otherwise $v,a_1,a_2,b$ induces a diamond, contradicting the assumption that
$G$ is a block graph.

Now, suppose conversely that there are $u,v \in I$ with degree at least 2.
Let $a_1,a_2$ be two vertices in $C$. From the paragraph above,
$a_1,a_2$ are neighbors of $u$ and of $v$.
Therefore, $u,v,a_1,a_2$ induces a diamond, contradicting the assumption that
$G$ is a block graph.

To prove the opposite direction, suppose that $G$ satisfied (1) and (2).
Suppose conversely that $G$ is not a block graph.
Then there is a set of vertices $X$ that induces a diamond.
Since $C$ is a clique and $I$ is an independent set, $|X\cap I|$ is
equal to either 1 or 2.
If $|X\cap I|=1$ then $G$ does not satisfy (1), and if $|X\cap I|=2$ then
$G$ does not satisfy (2), a contradiction.
Therefore $G$ is a block graph.
\end{proof}

IF $(G,k)$ is a yes instance, let $S$ be a solution for $(G,k)$.
By Lemma~\ref{lem:block}, there is at most one vertex in $I\setminus S$ with
degree at least 2 in $G-S$.
Denote this vertex, if it exists, by $v^*$.
For every $v \in I \setminus (\{v^*\} \cup S)$ we have that
$v$ has degree at most 1 in $G-S$.

The algorithm for SBVD goes over all possible choices for the vertex
$v^* \in I$.
Additionally, the algorithm also inspects the case in which no such vertex
exist.
For every choice of $v^*$, the algorithm deletes from $G$ all the vertices of
$C$ that are not adjacent to $v^*$ and decreases the value of $k$ by the number
of vertices deleted.
When the algorithm inspects the case when $v^*$ does not exists, the graph is
not modified.

For every choice of $v^*$,
the algorithm generates an instance $(\mathcal{F},k)$ of 3-Hitting Set as follows:
For every vertex $v \in I \setminus \{v^*\}$ that has at least two neighbors,
and for every two neighbors $a,b \in C$ of $v$, the algorithm adds
the set $\{v,a,b\}$ to $\mathcal{F}$.
The algorithm then uses the algorithm of
Wahlstr{\"o}m~\cite{wahlstrom2007algorithms} to solve the instance
$(\mathcal{F},k)$ in $O^*(2.076^k)$ time.
If $(\mathcal{F},k)$ is a yes instance of 3-Hitting Set then the algorithm
returns yes.
If all the constructed 3-Hitting Set instances, for all choices of $v^*$, are no
instances, the algorithm returns no.

\section{Algorithm for STVD}

Let $I_0$ be the set of all vertices in $I$ that have minimum degree
(namely, a vertex $u \in I$ is in $I_0$ if $|N(u)|\leq |N(v)|$ for every
$v\in I$).
We say that two vertices $u,v \in I$ are \emph{twins} if $N(u) = N(v)$.
Let $\twins{v}$ be a set containing $v$ and all the twins of $v$.
Recall that a split graph $G$ is a threshold graph if and only if $G$ does not
contain an induced $P_4$.
Note that an induced $P_4$ in $G$ must be of the form $u,a,b,v$ where
$u,v \in I$ and $a,b \in C$.

The algorithm for STVD is a branching algorithm.
At each step, the algorithm applies the first applicable rule from the rules
below.
The reduction rules of the algorithm are as follows.
\begin{rrule}
If $k\leq 0$ and $G$ is not a threshold graph, return `no'.
\end{rrule}

\begin{rrule}
If $G$ is an empty graph, return `yes'.
\end{rrule}

\begin{rrule}
If $v$ is a vertex such that there is no induced $P_4$ in $G$ that contains $v$,
delete $v$.\label{rrule:no-P4}
\end{rrule}

If Rule~(R\ref{rrule:no-P4}) cannot be applied we have that
for every $a\in C$ there is a vertex $v \in I$ such that $a \notin N(v)$.

We now describe the branching rules of the algorithm.
When we say that the algorithm branches on sets $S_1,\ldots,S_p$, we mean
that the algorithm is called recursively on the instances
$(G-S_1,k-|S_1|),\ldots,(G-S_p,k-|S_p|)$.

\begin{brule}
If there are non-twin vertices $u,v \in I$ such that $|N(u)| = |N(v)| = 1$,
branch on $N(u)$ and $N(v)$.\label{brule:degree-1a}
\end{brule}

To show the safeness of Rule~(\currentrule),
denote $N(u) = \{a\}$ and $N(v) = \{b\}$.
If $S$ is a solution for the instance $(G,k)$ then $S$ must contain at least one
vertex from the induced path $u,a,b,v$.
If $u \in S$ then $S' = (S \setminus \{u\}) \cup \{a\}$ is also a solution
(since every induced $P_4$ that contains $u$ also contains $a$).
Additionally, if $v \in S$ then $(S \setminus \{v\}) \cup \{b\}$ is also a
solution.
Therefore, there is a solution $S$ such that either $a \in S$ or $b \in S$.
Thus, Rule~(\currentrule) is safe.

The branching vector of Rule~(\currentrule) is $(1,1)$.

\begin{brule}
If there is a vertex $u \in I$ such that $|N(u)| = 1$, let $v \in I$
be a vertex such that $N(u) \not\subseteq N(v)$.
Branch on $\{v\}$, $N(u)$ and $N(v)$.\label{brule:degree-1}
\end{brule}

Note that the vertex $v$ exists since Rule~(R\ref{rrule:no-P4}) cannot be
applied.
To prove the safeness of Rule~(\currentrule), note that if $S$ is a solution
for the instance $(G,k)$ then either $u\in S$, $v\in S$,
$N(u) \subseteq S$, or $N(v) \subseteq S$.
If one of the last three cases occurs we are done.
Otherwise (if $u \in S$), $S' = (S \setminus \{u\}) \cup N(u)$ is also
a solution.
It follows that Rule~(\currentrule) is safe.

Since Rule~(B\ref{brule:degree-1a}) cannot be applied, $|N(v)| \geq 2$.
Therefore, the branching vector of Rule~(\currentrule) is at least $(1,1,2)$.

Note that if Rule~(\currentrule) cannot be applied, every vertex in $I$ has
degree at least~2.

\begin{brule}
If there are vertices $u,v \in I$ such that
$|N(u) \setminus N(v)| \geq 2$ and $|N(v) \setminus N(u)| \geq 2$,
branch on $\{u\}$, $\{v\}$, $N(u) \setminus N(v)$,
and $N(v) \setminus N(u)$.\label{brule:2-2}
\end{brule}

If $S$ is a solution for the instance $(G,k)$ then either
$u \in S$, $v\in S$, $N(u) \setminus N(v) \subseteq S$, or
$N(v) \setminus N(u) \subseteq S$
(If neither of the above cases hold, let
$a \in (N(u) \setminus N(v))\setminus S$ and
$b \in (N(v) \setminus N(u))\setminus S$.
Then, $u,a,b,v$ is an induced $P_4$ in $G-S$, a contradiction).
Therefore, Rule~(\currentrule) is safe.

The branching vector of Rule~(\currentrule) is at least $(1,1,2,2)$.

\begin{lemma}\label{lem:Nu-Nv}
If Rule~(\currentrule) cannot be applied and $u,v \in I$ are two
vertices such that $|N(u)| \leq |N(v)|$
then $|N(u) \setminus N(v)| \leq 1$.
\end{lemma}
\begin{proof}
Suppose conversely that $|N(u) \setminus N(v)| \geq 2$.
Then, $|N(v) \setminus N(u)| \geq |N(u) \setminus N(v)| \geq 2$.
Therefore, Rule~(\currentrule) can be applied on $u,v$, a contradiction.
\end{proof}

We now consider two cases.

\paragraph{Case 1}
In the first case, every two vertices in $I_0$ are twins.
The algorithm picks an arbitrary vertex $u \in I_0$ and vertices
$a_1,a_2 \in N(u)$.
Since Rule~(R\ref{rrule:no-P4}) cannot be applied, there is
a vertex $v_1 \in I$ such that $a_1 \notin N(v_1)$ and
a vertex $v_2 \in I$ such that $a_2 \notin N(v_2)$.
For $i = 1,2$ we have that $v_i \notin I_0$ since $v_i$ is not a twin of $u$.
Since $u \in I_0$, it follows that $|N(v_i)| > |N(u)|$.
Thus, $|N(v_i) \setminus N(u)| > |N(u) \setminus N(v_i)| \geq 1$.
By Lemma~\ref{lem:Nu-Nv} and the fact that $a_i \in N(u) \setminus N(v_i)$
we obtain that $ N(u) \setminus N(v_i) = \{a_i\}$ and thus
$N(u) \setminus \{a_i\} \subseteq N(v_i)$.
In particular, $a_2 \in N(v_1)$ and $a_1 \in N(v_2)$.
Note that this implies that $v_1 \neq v_2$.

\begin{lemma}
$|(N(v_1) \cap N(v_2)) \setminus N(u)| \geq 2$.\label{lem:Nv1Nv2}
\end{lemma}
\begin{proof}
Suppose without loss of generality that $|N(v_1)| \leq |N(v_2)|$.
By Lemma~\ref{lem:Nu-Nv} and the fact that $a_2 \in N(v_1) \setminus N(v_2)$
we have that % $N(v_1) \setminus N(v_2) = \{a_2\}$ and therefore
$N(v_1) \setminus N(u) \subseteq N(v_2)$.
Therefore, $(N(v_1) \cap N(v_2)) \setminus N(u) = N(v_2) \setminus N(u)$.
We have shown above that $|N(v_2) \setminus N(u)| \geq 2$.
\end{proof}

\begin{brule}
If Case~1 occurs and
$a_1,a_2 \in N(w)$ for every $w \in I \setminus \{v_1,v_2\}$,
branch on $\{u\}$, $(N(v_1) \cap N(v_2)) \setminus N(u)$,
and $\{v_1,v_2\}$.\label{brule:case1-1}
\end{brule}

We now prove the safeness of Rule~(\currentrule).
In order to delete the paths of the form $u,a_1,b,v_1$ or $u,a_2,b,v_2$
for some $b \in (N(v_1) \cap N(v_2)) \setminus N(u)$,
a solution $S$ must satisfy one of the following
(1) $u \in S$
(2) $(N(v_1) \cap N(v_2)) \setminus N(u) \subseteq S$, or
(3) $S$ contains at least one vertex from $\{a_1,v_1\}$ and at least one vertex
from $\{a_2,v_2\}$.
Suppose that $S$ is a solution that satisfies~(3).
Due to the assumption of Rule~(\currentrule) and the fact that $a_1 \in N(v_2)$,
we have that every vertex in $I \setminus \{v_1\}$ is adjacent to $a_1$.
Therefore, every induced $P_4$ that contains $a_1$ is of the form $v_1,x,a_1,y$.
Thus, if $v_1 \notin S$  then $S' = (S \setminus \{a_1\}) \cup \{v_1\}$
is also a solution.
Similarly, if $v_2 \notin S$ then $S' = (S \setminus \{a_2\}) \cup \{v_2\}$
is also a solution.
Therefore, if $(G,k)$ is a yes instance, there is a solution $S$ such that
either $S$ satisfies (1) or (2) above, or $\{v_1,v_2\} \subseteq S$.

By Lemma~\ref{lem:Nv1Nv2},
the branching vector of Rule~(\currentrule) is at least $(1,2,2)$.

\begin{brule}
If Case~1 occurs, let $w \in I \setminus \{v_1,v_2\}$ be a vertex such that
$\{a_1,a_2\} \not\subseteq N(w)$,
and without loss of generality assume that $w_1 \notin N(w)$.
Branch on $\{u\}$, $(N(v_1) \cap N(v_2)) \setminus N(u)$,
and on the sets in
$\{ \{a_1\}, \{v_1\} \cup (N(w) \setminus  N(u)), \{v_1,w\} \}
\circ \{ \{a_2\},\{v_2\} \}$.\label{brule:case1-2}
\end{brule}

We now show the safeness of Rule~(\currentrule).
In order to delete the induced paths of the form $u,a_1,b,v_1$ or $u,a_2,b,v_2$
for $b \in (N(v_1) \cap N(v_2)) \setminus N(u)$,
a solution $S$ must satisfy (1), (2), or (3) above.
Suppose that (1) is not satisfied (namely, $u \notin S$) and that (3) is
satisfied.
Additionally, suppose that $a_1 \notin S$.
Therefore, $v_1 \in S$ and $S$ contains at least one vertex from $\{a_2,v_2\}$.
In order to delete the induced paths of the form $u,a_1,c,w$ for
every $c \in N(w) \setminus N(u)$,
either $w \in S$ or $N(w)\setminus N(u) \subseteq S$.
We have that $w \notin I_0$ since $w$ is not a twin of $u$.
Since $u \in I_0$, it follows that $|N(w)| > |N(u)|$.
Thus, $|N(w) \setminus N(u)| > |N(u) \setminus N(w)| \geq 1$.
From the previous inequality, Lemma~\ref{lem:Nv1Nv2},
and the fact that $v_1,v_2,a_2 \notin N(w) \setminus N(u)$,
it follows that the branching vector of Rule~(\currentrule)
is at least $(1,2,2,4,3,2,4,3)$.

\paragraph{Case 2}
In the second case, there are non-twin vertices in $I_0$.
Suppose that $u_1,u_2 \in I_0$ are non-twin vertices,
where the choice of $u_1,u_2$ will be given later.
By Lemma~\ref{lem:Nu-Nv},
$|N(u_1) \setminus N(u_2)| = |N(u_2) \setminus N(u_1)| = 1$.
Denote $N(u_1) \setminus N(u_2) = \{a_1\}$ and
$N(u_2) \setminus N(u_1) = \{a_2\}$.
Let $I_1 = I_0  \setminus (\twins{u_1} \cup \twins{u_2})$.

\begin{lemma}\label{lem:sunflower}
If $I_1 \neq \emptyset$ then either
(1) for every $u \in I_1$,
$N(u)$ consists of $N(u_1) \cap N(u_2)$ plus an additional vertex
that is not in $\{a_1,a_2\}$,
or
(2) for every $u \in I_1$,
$N(u)$ consists of $a_1$, $a_2$, and all the vertices of $N(u_1) \cap N(u_2)$
except one vertex.
\end{lemma}
\begin{proof}
We first claim that every $u \in I_1$, $|N(u) \cap {a_1,a_2}|$ is either 0 or 2.
Suppose conversely that $N(u)$ contains exactly one vertex from $a_1,a_2$ and
without loss of generality, $a_2 \in N(u)$ and $a_1 \notin N(u)$.
By Lemma~\ref{lem:Nu-Nv} on $u,u_1$ we obtain that $N(u)$ contains all the
vertices in $N(u_1) \setminus \{a_1\} = N(u_2) \setminus \{a_2\}$.
Since we assumed that $a_2 \in N(u)$, we have that $N(u_2) \subseteq N(u)$.
From the fact that $|N(u_2)| = |N(u)|$ we obtain that $N(u_2) = N(u)$,
contradicting the assumption that $u$ is not a twin of $u_2$.
Therefore, $|N(u) \cap {a_1,a_2}|$ is either 0 or 2.

We first assume that there is no vertex $u \in I_1$ such that
$a_1,a_2 \in N(u)$.
From the claim above we have that $a_1 \notin N(u)$.
By Lemma~\ref{lem:Nu-Nv} on $u,u_1$, $N(u)$ contains all the vertices in
$N(u_1) \setminus \{a_1\} = N(u_1)\cap N(u_2)$ plus an additional vertex
that is not in $\{a_1,a_2\}$.

Now suppose that there is a vertex $u_3 \in I_1$ such that $a_1,a_2 \in N(u_3)$.
Consider some $u \in I_1$. We claim that $a_1 \in N(u)$.
Suppose conversely that $a_1 \notin N(u)$.
From the claim above, $a_2 \notin N(u)$.
Therefore, $a_1,a_2 \in N(u_3)\setminus N(u)$,
contradicting Lemma~\ref{lem:Nu-Nv}.
Thus, $a_1 \in N(u)$. From the claim above and Lemma~\ref{lem:Nu-Nv} we conclude
that $N(u)$ contains $a_2$ and all the vertices in $N(u_1) \cap N(u_2)$ except
one vertex.
\end{proof}
If $I_1 = \emptyset$ or the first case of Lemma~\ref{lem:sunflower} occurs,
we say that the vertices of $I_0$ form a \emph{sunflower}.
Note that if the vertices of $I_0$ do not form a sunflower,
for every vertex $a \in N(u_1) \cup N(u_2)$ there are non-twin vertices
$u,u' \in I_0$ that are adjacent to $a$.

\begin{brule}
If there are non-twin vertices $u_1,u_2 \in I_0$ such that $a_1,a_2 \in N(w)$
for every $w \in I \setminus I_0$, then suppose without loss of generality that
$|\twins{u_1}| \leq |\twins{u_2}|$.
Branch on $\twins{u_1}$ and $\{a_1\}$.\label{brule:case2-1}
\end{brule}

To prove the safeness of Rule~(\currentrule), suppose that $(G,k)$ is a yes
instance and let $S$ be a solution.
If $a_1 \in S$ or $\twins{u_1} \subseteq S$ we are done, so suppose that
that $a_1 \notin S$ and $\twins{u_1} \not\subseteq S$.
We can assume that $S \cap \twins{u_1} = \emptyset$ (otherwise,
$S' = S \setminus \twins{u_1}$ is also a solution).
Since $u'_1,a_1,a_2,u'_2$ is an induced path for every $u'_1 \in \twins{u_1}$
and $u'_2 \in \twins{u_2}$, either $a_2 \in S$ or $\twins{u_2} \subseteq S$.
Note that we can assume that if $S$ contains at least one vertex from
$\twins{u_2}$ then it contains all the vertices of $\twins{u_2}$.
%(if $S$ contains some but not all of the vertices in $\twins{u_2}$ then
%$S' = S \setminus \twins{u_2}$ is also a solution).
Define a set $S'$ by taking the vertices in
$S \setminus (\{a_2\} \cup \twins{u_2})$.
Additionally, if $a_2 \in S$, add $a_1$ to $S'$, and
if $\twins{u_2} \subseteq S$, add $\twins{u_1}$ to $S'$.
We now show that $S'$ is also a solution.
Since $|\twins{u_1}| \leq |\twins{u_2}|$, we have that $|S'| \leq |S| \leq k$.
Suppose conversely that $G-S'$ contains an induced $P_4$ and denote this path
by $P'$.
Create a path $P$ by taking $P'$ and performing the following steps:
(1) If $a_1$ is in $P'$, replace it with $a_2$.
(2) If $a_2$ is in $P'$, replace it with $a_1$.
(3) If $P'$ contains a vertex $u'_1 \in \twins{u_1}$, replace it with $u_2$.
(4) If $P'$ contains a vertex $u'_2 \in \twins{u_2}$, replace it with $u_1$.
Recall that $a_1,a_2 \in N(w)$ for every $w \in I \setminus I_0$.
Additionally, for every $u \in I_1$, $a_1,a_2 \notin N(u)$ if
the vertices of $I_0$ form a sunflower, and $a_1,a_2 \in N(u)$ otherwise.
Therefore, for every two vertices $x',y'$ in $P'$ and the corresponding vertices
$x,y$ in $P$, we have that $(x,y)$ is an edge if and only if $(x',y')$ is
an edge.
It follows that $P$ is also an induced path in $G$.
From the assumptions that $a_1 \notin S$ and $\twins{u_1} \cap S = \emptyset$
and from the definition of $S'$ we have that $S$ does not contain a vertex of
$P$.
This contradicts the assumption that $S$ is a solution.
Therefore, $S'$ is a solution. The solution $S'$ contains either $\twins{u_1}$
or $\{a_1\}$, and therefore Rule~(\currentrule) is safe.

The branching vector of Rule~(\currentrule) is at least $(1,1)$.

Now suppose that Rule~(B\ref{brule:case2-1}) cannot be applied.
We choose non-twin vertices $u_1,u_2 \in I_0$,
a vertex $a \in N(u_1) \cap N(u_2)$, and a vertex $v \in I \setminus I_0$ that
is not adjacent to $a$ as follows.
\begin{enumerate}
\item
If the vertices of $I_0$ form a sunflower,
pick arbitrary non-twin vertices $u_1,u_2 \in I_0$.
Pick $a \in N(u_1) \cap N(u_2)$.
Since Rule~(R\ref{rrule:no-P4}) cannot be applied, there is a vertex $v\in I$
such that $a \notin N(v)$.
Since $a \in N(u)$ for every $u \in I_0$
(as the vertices of $I_0$ form a sunflower)
it follows that $v \in I \setminus I_0$.
\item
Otherwise, since Rule~(B\ref{brule:case2-1}) cannot be applied,
there is a vertex $v \in I \setminus I_0$ such that
$\bigcup_{u\in I_0} N(u) \not\subseteq N(v)$.
Pick $a \in (\bigcup_{u\in I_0} N(u)) \setminus N(v)$.
Since the vertices of $I_0$ do not form a sunflower,
there are non-twin vertices $u_1,u_2 \in I_0$ such that
$a \in N(u_1) \cap N(u_2)$.
\end{enumerate}
Since Rule~(B\ref{brule:case2-1}) cannot be applied, there is a vertex
$w \in I \setminus I_0$ such that, without loss of generality,
$a_1 \notin N(w)$.

\begin{brule}
Branch on $\{u_1\}$, $(N(v) \cap N(w)) \setminus N(u_1)$,
and on the sets in
$\{ \{a\},\{v\} \} \circ \{ \{a_1\},\{w,a_2\},\{w,u_2\} \}$.
\end{brule}

The proof of the safeness of Rule~(\currentrule) is similar to the proof
for Rule~(B\ref{brule:case1-2}).
To bound the branching vector of Rule~(\currentrule) we use the following lemma.

\begin{lemma}
$|(N(v) \cap N(w)) \setminus N(u_1)| \geq 2$.\label{lem:NvNw}
\end{lemma}
\begin{proof}
Since $|N(v)| > |N(u_1)|$, we have that
$|N(v) \setminus N(u_1)| > |N(u_1) \setminus N(v)| \geq 1$.
Similarly,
$|N(w) \setminus N(u_1)| > |N(u_1) \setminus N(w)| \geq 1$.
By Lemma~\ref{lem:Nu-Nv} and the fact that $a_1 \in N(u_1) \setminus N(w)$
we have that $a \in N(w)$.

We consider two cases.
If $|N(v)| > |N(w)|$ then by Lemma~\ref{lem:Nu-Nv} and the fact that
$a \in N(w) \setminus N(v)$ we have that
$N(w) \setminus N(u_1) \subseteq N(w) \setminus \{a\} \subseteq N(v)$.
Therefore, $(N(v)\cap N(w)) \setminus N(u_1) = N(w) \setminus N(u_1)$
and the lemma follows since $|N(w) \setminus N(u_1)| \geq 2$.

If $|N(v)| \leq |N(w)|$ then by Lemma~\ref{lem:Nu-Nv} and the fact that
$a_1 \in N(v)\setminus N(w)$ we have that
$N(v) \setminus N(u_1) \subseteq N(v) \setminus \{a_1\} \subseteq N(w)$.
Therefore, $(N(v)\cap N(w)) \setminus N(u_1) = N(v)\setminus N(u_1)$
and the lemma follows since $|N(v) \setminus N(u_1)| \geq 2$.
\end{proof}

By Lemma~\ref{lem:NvNw},
the branching vector of Rule~(\currentrule) is at least $(1,2,2,2,3,3,3,3)$.

The rule with largest branching number is Rule~(B\ref{brule:2-2}) and
its branching number is at most 2.733.
Therefore, the running time of the algorithm is $O^*(2.733^k)$.

\bibliographystyle{abbrv}
\bibliography{split}

\end{document}